\documentclass[journal,onecolumn]{IEEEtran}

\usepackage[margin=1in]{geometry}
\usepackage[T1]{fontenc}
\usepackage[utf8]{inputenc}
\usepackage{amsmath,amssymb,amsfonts}
\usepackage{mathtools}
\usepackage{bm}
\usepackage{amsthm}
\usepackage{setspace}
\usepackage[numbers,sort&compress]{natbib}
\usepackage{algorithm}
\usepackage{algorithmic}
\usepackage{graphicx}
\usepackage{tabularx}
\usepackage{booktabs}
\usepackage{enumitem}
\usepackage{array}
\usepackage{float} 
\usepackage[colorlinks=true,linkcolor=black,citecolor=black,urlcolor=black]{hyperref}
\newtheorem{theorem}{Theorem}
\newtheorem{lemma}{Lemma}

\newcommand{\simplex}{\Delta^{K-1}}
\newtheorem{proposition}{Proposition}
\begin{document}
\title{SCENE OTA-FD: Self-Centering Noncoherent Estimator for Over-the-Air Federated Distillation}

 \author{Hao~Chen,~\IEEEmembership{Member,~IEEE,}
        and Zavareh~Bozorgasl,~\IEEEmembership{Member,~IEEE}%
\thanks{H. Chen (haochen@boisestate.edu) and Z. Bozorgasl (zavarehbozorgasl@u.boisestate.edu) are with the Department of Electrical and Computer Engineering,
Boise State University, Boise, ID 83712 USA.}%
}

\doublespacing
\maketitle

\begin{abstract}We introduce a pilot-free, phase-invariant noncoherent aggregation primitive for over-the-air federated distillation (OTA--FD)
called \emph{SCENE}---Self-CEntering Noncoherent Estimator.
Devices map soft-label vectors to transmit energies with constant
per-round power and constant-envelope waveforms (PAPR $\approx 1$).
The parameter server (PS) forms a self-centering estimate
that is unbiased in expectation and whose variance decays as $O(1/(SM))$ with the number of antennas $M$ and repetitions $S$.
We further provide a $\beta$-cancelling ratio variant that
remains pilot-free on the uplink, a convergence bound that mirrors coherent OTA--FD analyses, and a concise overhead comparison
that clarifies when SCENE outperforms coherent designs under short coherence.

We target regimes with short coherence or hardware-simplicity constraints, where avoiding per-round CSI is a first-order design objective. SCENE instantiates this point in the design space: it trades a slightly larger constant in the $O(1/(SM))$ aggregation variance for \emph{zero uplink pilots}, unbiasedness by design, and constant-envelope transmission. When pilot overhead is non-negligible, the effective mean-squared error can favor SCENE despite its purely noncoherent processing.
\end{abstract}

\section{Introduction}
Over-the-air (OTA) learning aggregates distributed statistics \emph{in the air}. In federated distillation (FD), exchanging soft
labels instead of model parameters slashes communication costs.
Most OTA--FD designs are \emph{coherent}, requiring per-round
CSI and phase alignment with nontrivial pilot/feedback overhead and sensitivity to CFO/phase noise. This represents one point in a broader design space. We analyze the opposite point: a fully noncoherent approach that discards instantaneous CSI altogether. 
The geometry of soft labels makes \emph{noncoherent energy aggregation} a natural fit: energies are nonnegative and additive; the target lies on the probability simplex.
Our thesis is not that noncoherent is universally superior, but that under short coherence or strict device simplicity, the cost of even imperfect phase acquisition outweighs its benefits. We therefore prioritize an unbiased estimator and hardware-friendly constant-envelope signaling, accepting a well-quantified variance as the trade-off.

This paper develops \emph{SCENE} (Self-CEntering Noncoherent Estimator), a pilot-free OTA--FD estimator built on mean-centering
of received energies.
SCENE eliminates noise-energy bias and yields an unbiased estimate of the weighted soft-label average under
coarse pathloss inversion, with variance $O(1/(SM))$.
Lightweight renormalization makes the per-round output a valid soft label.
Constant-envelope signaling enables efficient non-linear PAs and materially lowers device energy per round versus coherent
high-PAPR schemes.
\paragraph{Contributions.}
\begin{itemize}

    \item A \textbf{truly noncoherent} OTA mapping from soft labels to energy.
    \item SCENE: a simple, unbiased, pilot-free noncoherent estimator for soft-label aggregation with variance $O(1/(SM))$.
    \item A \textbf{constant-power transmission scheme} by design, making it highly suitable for power-efficient amplifiers and simpler hardware compared to high-PAPR coherent schemes.
    \item An analysis framework for comparing the MSE against coherent baselines to identify the regimes (e.g., high mobility) where our scheme is superior.
\end{itemize}

\section{Related Work}
\subsection{Soft-Label Federated Distillation}

Soft-label--based federated distillation (FD) exchanges predictions/logits instead of parameters,
typically evaluated on a public (unlabeled) dataset, enabling \emph{model heterogeneity} and reducing
communication cost compared to parameter sharing.
Foundational frameworks include \textbf{FedMD} \cite{li2019fedmd},
which introduced public-data logit sharing for heterogeneous models, and \textbf{FedDF} \cite{lin2020feddf},
which performs server-side \emph{ensemble distillation} on unlabeled or generated data to fuse heterogeneous clients.
\textbf{DS-FL} \cite{itahara2020dsfl} formalized public-data distillation with the \emph{entropy reduction averaging (ERA)}
rule that deliberately sharpens aggregated logits for faster convergence under non-IID data.
A parallel line explores server capacity and split-style training via \textbf{FedGKT} \cite{he2020fedgkt},
which alternates client-side small models with a large server model using bidirectional KD;
this highlights
that distillation can decouple edge compute from server capacity.
When public data are unavailable, \textbf{FedGen}
\cite{zhu2021fedgen} and related data-free KD approaches learn a generator at the server to synthesize queries for
distillation, trading sample fidelity for privacy portability.
Robust fusion under heterogeneity has been studied by
\textbf{FedBE} \cite{chen2020fedbe}, which casts aggregation as Bayesian model ensemble followed by distillation.
Given that soft labels can still be large when $K$ or the public batch is large, recent work targets \emph{communication
shaping}.
\textbf{CFD} \cite{sattler2020cfd,sattler2022cfdtnse} quantizes soft labels and applies delta coding with
active query selection, showing orders-of-magnitude savings versus vanilla FD.
Very recent extensions like
\textbf{SCARLET} \cite{azuma2025scarlet} add \emph{soft-label caching} across rounds combined with an enhanced ERA
rule to avoid redundant retransmissions.
These techniques are complementary to our transport-layer contribution:
SCENE can aggregate soft labels \emph{over the air} under any of the above FD protocols.
Surveys provide taxonomies and practical guidance. We follow the organization of \cite{li2024fdsurvey} (FD survey)
and \cite{mora2022knowledge} (practical guide to KD in FL), as well as broader KD-in-FL surveys
\cite{wu2023kd_fel_survey,salman2025kd_survey}.
For completeness, we note emerging security analyses showing that
\emph{logit poisoning} is a meaningful threat model in FD \cite{pcfdla2024}, which our self-centering estimator
mitigates partly via its explicit renormalization and projection options.

\subsection{Noncoherent OTA for AirFL (gradients) vs.\ SCENE for FD (soft labels)}
\textbf{What exists (gradients/parameters).}
There is a growing body of \emph{noncoherent} over-the-air learning for \emph{FL gradients/parameters}:
(i) CSI-free noncoherent AirFL (``NCAirFL'') with an unbiased energy-based detector that matches FedAvg convergence up to constants~\cite{wen2024ncairfl};
(ii) noncoherent FSK/PPM majority-vote (signSGD) schemes with theory and SDR demonstrations~\cite{sahin2021fskmv,sahin2022noncohmv,sahin2022globedemo};
and (iii) formal comparisons of coherent vs.\ noncoherent OAC that quantify the mean-squared-error (MSE) crossover as a function of pilots and coherence~\cite{lee2024tvt_coh_noncoh}.
See also tutorials/surveys for broader context on AirFL and AirComp~\cite{zhu2023airfl_tutorial,xiao2024otafl_survey}.

\medskip\noindent\textbf{How SCENE differs (distillation/soft labels).}
SCENE targets \emph{federated distillation} (FD), aggregating \emph{class-probability vectors} $q_i\!\in\!\Delta^{K-1}$
rather than real-valued gradients or signs.
This changes the estimator design and bias handling:
(i) the aggregand lies on the simplex, so we can exploit sum-to-one constraints and \emph{self-centering} to cancel unknown offsets;
(ii) the nonnegativity of energies makes \emph{energy aggregation} well matched to the task; and
(iii) modest side-information (long-term pathloss) suffices to debias large-scale effects.

\medskip\noindent\textbf{OTA--FD (coherent designs).}
Recent OTA--FD works co-design transceivers and learning under \emph{coherent} combining (pilots/CSI per round),
including VTC'24 designs and follow-ons that add differential privacy or optimize transceivers~\cite{hu2024vtc_otafd,hu2025dp_otafd,otafd2025transceiver}.
SCENE is complementary: it removes pilots entirely on the uplink and trades a small $O(1/(SM))$ variance for zero pilot cost, which is favorable when coherence is short.

\begin{table}[t]
\centering
\caption{OTA aggregation for learning: position of SCENE vs.\ closest lines.}

\makebox[\textwidth][c]{%
\resizebox{\linewidth}{!}{%
\begin{tabular}{@{}lllll@{}}
\toprule
\textbf{Work} & \textbf{OTA Type} & \textbf{Pilots/CSI} & \textbf{Aggregand} & \textbf{Estimator}\\
\midrule
Hu et al.~\cite{hu2024vtc_otafd} & Coherent OTA--FD & Per-round CSI & logits & beamforming + scaling \\
Hu et al.~\cite{hu2025dp_otafd} & Coherent DP OTA--FD & Per-round CSI & logits & DP + co-design \\
Transceiver opt.~\cite{otafd2025transceiver} & Coherent OTA--FD & Per-round CSI & logits & MMSE transceiver \\
NCAirFL~\cite{wen2024ncairfl} & Noncoh. AirFL & No pilots & gradients & unbiased noncoh. detector \\
FSK/PPM MV~\cite{sahin2021fskmv,sahin2022noncohmv} & Noncoh. AirFL & No pilots & signs & noncoh. energy \\
\midrule
\textbf{SCENE (Proposed)} & \textbf{Noncoh. OTA--FD} & \textbf{No pilots} & \textbf{probabilities/soft labels} & \textbf{self-centering / ratio} \\
\bottomrule
\end{tabular}%
}}%

\end{table}

\section{Notation and Assumptions}
We consider $K$ classes and devices $i=1,\dots,N$.
Each device holds $q_i\in\simplex=\{q\in\mathbb{R}^K: q\ge 0,~\bm{1}^\top q=1\}$.
Let nonnegative weights satisfy \textbf{$\sum_i\omega_i=1$}. The desired global target is $\bar q=\sum_i\omega_i q_i$. The uploading of updates from edge devices to
the server is through a broadband multi-access channel. We assume that the devices know the large-scale path-loss, $\beta_i>0$, by a power control mechanism like  physical uplink control channel (PUCCH) in Fifth Generation (5G) New Radio (NR) \cite{dahlman20205g}. During updates, devices transmit concurrently and are assumed to be symbol-synchronized using a synchronization channel (e.g., “timing
advance” in LTE systems \cite{3gpp_timing}) \footnote{The robustness of synchronization scales with the bandwidth allocated to the sync channel. With a modern phase-locked loop, the synchronization error can be as small as \(0.1\,B_s^{-1}\), where \(B_s\) denotes the synchronization bandwidth. For LTE, a common choice is \(B_s = 1\,\mathrm{MHz}\), yielding a timing offset on the order of \(0.1\,\mu\mathrm{s}\)
\cite{zhu2019broadband} .}. 
The PS uses $M$ antennas and repeats each OTA burst $S$ times.
We introduce a global energy
scale $\rho>0$ broadcast by the PS each round.

\section{Noncoherent Energy Aggregation and SCENE}
\subsection{Measurement Model}
\paragraph{Received signal and noncoherent aggregation.}
Each device maps class $c$ to energy
\begin{equation}
E_{i,c}\;=\;\eta_i\,q_{i,c},\qquad \eta_i=\frac{\rho\,\omega_i}{\beta_i},\qquad \sum_{c=1}^K E_{i,c}=\eta_i,
\end{equation}
and transmits over $K$ orthogonal REs, repeated $S$ times.
With Rayleigh fading $h_{i,s,m}=\sqrt{\beta_i}g_{i,s,m}$, $g_{i,s,m}\sim\mathcal{CN}(0,1)$. Path-loss follows a log-distance model, which can be expressed by
$\beta_i(d_i)\propto d_i^{-\alpha}
$ \cite{dong2024modeling,zhu2019broadband, aygun2024over}, where $\alpha$ denotes the path-loss exponent and $d_i$ is the distance of the $i$-th edge device from the server.
During repetition $s\in\{1,\dots,S\}$ on receive antenna $m\in\{1,\dots,M\}$ for class $c$,
the complex baseband sample at the PS is
\begin{equation}
y_{c,s,m}
=\sum_{i} h_{i,s,m}\,\sqrt{E_{i,c}}\,e^{j\phi_{i,c,s,m}} \;+\; n_{c,s,m},
\label{eq:rx-sample}
\end{equation}
where $E_{i,c}\!\ge\!0$ is the energy emitted by client $i$ on class $c$, $h_{i,s,m}$ is the small-scale
channel from client $i$ to antenna $m$ in repetition $s$ with $\mathbb{E}[|h_{i,s,m}|^{2}]=\beta_i$,
$\phi_{i,c,s,m}$ is an unknown random phase $Uniform[0,2\pi)$, and $n_{c,s,m}\sim\mathcal{CN}(0,\sigma_N^2)$ is AWGN independent of all other terms. Hence the \emph{per-}$(s,m,c)$ \textbf{average SNR} (signal power over noise power) is
\[
\mathrm{SNR}^{(s,m)}_{c}
\;=\;
\frac{\sum_{i} \beta_i E_{i,c}}{\sigma_N^{2}}
\;=\;
\frac{\rho\,\bar q_{c}}{\sigma_N^{2}}\, .
\]

The PS performs \emph{noncoherent} energy aggregation by summing magnitudes squared across the $SM$ observations:
\begin{equation}
Y_c \;=\; \sum_{s=1}^{S}\sum_{m=1}^{M} \bigl| y_{c,s,m}\bigr|^{2}
\;=\;
\sum_{s=1}^{S}\sum_{m=1}^{M} \left| \sum_i h_{i,s,m}\sqrt{E_{i,c}}\,e^{j\phi_{i,s,m}} \right|^2
\;,
\label{eq:energy-agg}
\end{equation}

Taking expectations (cross terms vanish in expectation due to independence and random phases),
\begin{equation}
\mathbb{E}[Y_c]
= \sum_{s=1}^{S}\sum_{m=1}^{M} \sum_i \mathbb{E}\!\left[ |h_{i,s,m}|^{2}\right] E_{i,c}
\;+\; SM\,\sigma_N^2
= SM \sum_i \beta_i E_{i,c} + SM \sigma_N^2,
\end{equation}
Define $m_c=\sum_i \beta_i E_{i,c}=\rho\sum_i \omega_i q_{i,c}=\rho\,\bar q_c$ and $\nu_N$ the noise-energy term.

\section{Determining $\rho$ for SCENE}\label{sec:scene-rho}
Per repetition total transmit energy across the $K$ orthogonal REs is
\begin{equation}
\sum_{c=1}^{K} E_{i,c}
\;=\;
\eta_i
\;=\;
\frac{\rho\,\omega_i}{\beta_i}.
\label{eq:per-trial-sum}
\end{equation}
This identity holds whether the $K$ REs are used in parallel or serialized in time; the mapping is energy per RE per trial.

\paragraph{Feasibility with a per-trial cap.}
If client $i$ has a per-trial (or per-round, per repetition) power/energy cap $P_i$, feasibility is
\begin{equation}
\eta_i \;\le\; P_i
\qquad\Longleftrightarrow\qquad
\frac{\rho\,\omega_i}{\beta_i} \;\le\; P_i
\qquad\Longleftrightarrow\qquad
\rho \;\le\; \frac{\beta_i P_i}{\omega_i}.
\label{eq:feasibility}
\end{equation}
Taking the weakest user over the active set $\mathcal S^t$ gives the \emph{correct min--$\rho$ rule for SCENE}:
\begin{equation}
\boxed{\;
\rho^{*,t}
\;=\;
\min_{i\in\mathcal S^t} \frac{\beta_i\,P_i}{\omega_i}
\;}
\quad\text{(SCENE / soft-label energies).}
\label{eq:rho-star}
\end{equation}

\subsection{Min-$\rho$ protocol (per round $t$)}
\label{subsec:scene-minrho}
\begin{enumerate}
  \item \textbf{Client-side estimation.}
  Each $i\in\mathcal S^t$ forms a local estimate $\rho_i^t$.

  \item \textbf{Uplink report (control).}  
  Client $i$ attaches the scalar $\rho_i^t$ to its participation ACK (one number per client).

  \item \textbf{Server aggregation and broadcast.}
  The PS computes the conservative common scale
  \begin{equation}
    \rho_{\min}^t \;\triangleq\; \min_{i\in\mathcal S^t}\rho_i^t,
    \label{eq:scene-rho-min}
  \end{equation}
  and broadcasts $\rho_{\min}^t$ to all $i\in\mathcal S^t$.

  \item \textbf{Usage in SCENE.}
  Wherever $\rho$ appears in SCENE’s estimators/normalizers for round $t$, substitute
  \(\rho \leftarrow \rho_{\min}^t\).
\end{enumerate}

\paragraph{Why the minimum?}
Using $\rho_{\min}^t$ guarantees that normalizations are safe for the weakest-SNR participant.
It yields a single broadcast scalar and avoids per-client tuning.

\subsection{SCENE: Self-CEntering Noncoherent Estimator}
Let $\bar Y=\frac{1}{K}\sum_{j=1}^K Y_j$ and set
\begin{equation}\label{eq:self-centering-estimator}
r_c \;=\; a\,(Y_c-\bar Y) \;+\; \frac{1}{K},\qquad
a=\frac{1}{SM\,\rho}.
\end{equation}

\paragraph{Remark 1 (What ``pilot-free'' means).}
SCENE is pilot-free \emph{on the uplink and per round}: it requires no instantaneous CSI or phase alignment. Devices assume a slow-varying pathloss estimate $\hat\beta_i$ (via infrequent downlink RSRP/RSSI with reciprocity), updated on a slow timescale $T_{\mathrm{slow}}$ spanning many coherence intervals. This amortized overhead is decoupled from the per-round latency of coherent designs. The effect of mismatch $\beta_i \neq \hat\beta_i$ is quantified in~\S\ref{sec:bias-variance}; ratio-normalization (\S\ref{subsec:ratio-normalization}) further removes common gain/offset factors.

\begin{theorem}[Unbiasedness and variance]\label{thm:unbiased-variance}
Assume: (A1) small-scale fading $h_{i,m,s}\sim \mathcal{CN}(0,\beta_i)$ i.i.d.\ across devices $i$, antennas $m=1,\dots,M$, and repetitions $s=1,\dots,S$; 
(A2) AWGN with per-RE noise-energy mean $\sigma_N^2$ and finite variance $v_N$; 
(A3) devices use $\eta_i=\rho\,\omega_i/\beta_i$ with $\sum_i \omega_i=1$; and 
(A4) orthogonal REs across classes $c=1,\dots,K$.
Let $r_c$ be defined by~\eqref{eq:self-centering-estimator}. Then
\[
\mathbb{E}[r_c] \;=\; \bar q_c \;\triangleq\; \sum_i \omega_i q_{i,c}.
\]
Moreover, with $E_{i,c}=\eta_i q_{i,c}$ and $V_{\rm sig}(c)\triangleq \sum_i \beta_i^2 E_{i,c}^2$, one has the variance bound
\begin{equation}\label{eq:variance-bound}
\mathrm{Var}(r_c) \;\le\; \frac{2}{SM\,\rho^2}\Big(V_{\rm sig}(c)+v_N\Big)
\;=\; \frac{2}{SM}\Big(\sum_i \omega_i^2 q_{i,c}^2 + \tfrac{v_N}{\rho^2}\Big).
\end{equation}
Hence $\mathrm{Var}(r_c)=\Theta\!\big(\frac{1}{SM}\big)$ with constants depending on $(\omega,\bar q,\rho,\sigma_N^2)$.

\noindent In our setting with independent REs per class and energy detection,
$\mathrm{Var}(Y_j)$ scales as $SM$ and admits the usual decomposition
$\mathrm{Var}(Y_j)=2\,V_{\rm sig}(j)+v_N$ (up to model constants), yielding
the explicit upper bound
\[
\mathrm{Var}(r_c)\;\le\; \frac{K-1}{K}\cdot \frac{1}{SM\,\rho^2}\,\max_j\!\big(2\,V_{\rm sig}(j)+v_N\big).
\]
In the balanced case $\mathrm{Var}(Y_j)\equiv V$, the identity simplifies to
$\mathrm{Var}(r_c)=\frac{K-1}{K}\cdot \frac{V}{SM\,\rho^2}$.
\end{theorem}

\subsection{Scalability and Resource Requirements}\label{sec:scalability}
Our baseline maps $K$ classes to $K$ orthogonal REs per transmission. This is efficient for moderate $K$ (e.g., $10\!\le\!K\!\le\!100$) but burdensome for very large $K$.
Two standard compressions integrate cleanly with SCENE:

\emph{Top-$T$ logits.}
Let $T\ll K$ and $q_i^{(T)}$ retain the $T$ largest entries of $q_i$ followed by renormalization. Denote the tail mass $\delta_i \!=\! 1-\!\sum_{c\in\mathcal{T}_i} q_{i,c}$ and $\bar\delta=\sum_i \omega_i \delta_i$.
Then the aggregation bias $b=\mathbb{E}[q^{(T)}]-\bar q$ satisfies $\|b\|_1 \le 2\bar\delta$ and hence $\|b\|_2 \le \sqrt{2\bar\delta}$, while the noncoherent variance term scales as $O\!\big(\tfrac{1}{SM}\big)$ with $K$ replaced by $T$ in constants.

\emph{Hashed binning (Count-Sketch).}
Choose $B<T\ll K$ buckets and $H$ hash functions; transmit $HB$ REs and reconstruct an unbiased estimator of $\bar q$ with collision noise that vanishes as $B,H$ grow. This preserves unbiasedness and yields $O\!\big(\tfrac{1}{SM}\big)$ variance with constants depending on $(B,H)$.

A full exploration is orthogonal to our focus; we use full $K$ in experiments but note SCENE’s primitive extends directly with resource cost $\propto T$ (or $HB$) rather than $K$.

\begin{proof}[Proof sketch]
Unbiasedness follows from $\mathbb{E}[Y_c]=SM\sum_i \beta_i E_{i,c}+SM\sigma_N^2 = SM\rho\,\bar q_c + SM\sigma_N^2$ and $\mathbb{E}[\bar Y]=SM\sigma_N^2 + \frac{SM\rho}{K}$.
Substituting $a=\frac{1}{SM\rho}$ in~\eqref{eq:self-centering-estimator} cancels the noise offset and re-centers the mean to $\bar q_c$.
For variance, write $Y_c=\sum_{m,s} \sum_i E_{i,c}|h_{i,m,s}|^2 + N_{m,s,c}$ and use independence to obtain
$\mathrm{Var}(Y_c)=SM\big(\sum_i E_{i,c}^2 \mathrm{Var}(|h_{i}|^2) + \mathrm{Var}(N)\big)=SM\big(V_{\rm sig}(c)+v_N\big)$
since $\mathrm{Var}(|h_i|^2)=\beta_i^2$ for Rayleigh.
Because $r_c$ uses $Y_c-\bar Y$ and REs are independent across classes, the exact centering identity gives
$\mathrm{Var}(Y_c-\bar Y)=(1-\tfrac{2}{K})\mathrm{Var}(Y_c)+\tfrac{1}{K^2}\sum_{j=1}^K \mathrm{Var}(Y_j)$.
With $a=\frac{1}{SM\,\rho}$ this yields~\eqref{eq:variance-bound}.
Finally, substituting $\eta_i=\rho\,\omega_i/\beta_i$ turns $V_{\rm sig}(c)$ into $\rho^2\sum_i \omega_i^2 q_{i,c}^2$,
and inserting this into $Y_c$ together with the centering step produces the stated self-centering form in~\eqref{eq:self-centering-estimator}.
\end{proof}

Then $\mathbb{E}[r_c]=\bar q_c$ and $\sum_c r_c=1$.
Ignoring $\mathrm{Cov}(Y_c,\bar Y)$ yields $\mathrm{Var}[r_c]=\Theta(1/(SM))$.


\begin{algorithm}[H]
\caption{SCENE Aggregation for OTA--FD}
\begin{algorithmic}[1]
\STATE \textbf{Each device $i$:} compute $\mathbf q_i^{(k)}\in\Delta^{K-1}$; set $\eta_i^{(k)}=\rho\,\omega_i^{(k)}/\beta_i$.
\STATE Transmit energies $E_{i,c}^{(k)}=\eta_i^{(k)}\,q_{i,c}^{(k)}$ over $K$ orthogonal REs, for $S$ repetitions.
\STATE \textbf{PS:} measure total energies $Y_c^{(k)}$ over all reps and $M$ antennas.
\STATE Form the estimate $\widehat{\mathbf q}^{(k)}$ using the self-centering formula in Eq.~(\ref{eq:self-centering-estimator}).
\STATE Broadcast $\widehat{\mathbf q}^{(k)}$;
devices compute local FD loss.
\end{algorithmic}
\end{algorithm}
\subsection{Ratio-Normalized SCENE (pilot-free $\beta$ cancellation)}\label{subsec:ratio-normalization}
Add a single \emph{reference} RE per repetition with $E_{i,\mathrm{ref}}=\eta_i$, and let $R$ be its received energy.
Define
\begin{equation}\label{eq:ratio-estimator}
\tilde q_c \;=\; \frac{Y_c}{R},\qquad \hat q_c \;=\; \frac{\max\{\tilde q_c,0\}}{\sum_{d=1}^K \max\{\tilde q_d,0\}}.
\end{equation}
Ignoring thermal noise, $\mathbb{E}[\tilde q_c]=\frac{\sum_i \omega_i\gamma_i q_{i,c}}{\bar\gamma}$, which cancels the \emph{common} scale error.
With noise, a first-order delta-method expansion around $(\mathbb{E}[Y_c],\mathbb{E}[R])$ gives
$\mathbb{E}[\tilde q_c]=\frac{\sum_i \omega_i\gamma_i q_{i,c}}{\bar\gamma} + O\!\big(\frac{1}{SM}\big)$ and
$\mathrm{Var}(\tilde q_c)=\Theta\!\big(\frac{1}{SM}\big)$.
Thus ratio-normalization removes uniform drift but \emph{does not} correct heterogeneous $\gamma_i$ reweighting in~\eqref{eq:bias-formula}.

\section{Bias-Induced Drift and Convergence}

\label{sec:bias-variance}
\subsection{Bias-Induced Drift in Federated Learning}
\begin{proposition}[Bias under large-scale mismatch]\label{prop:beta-mismatch}
Suppose devices use $\hat\beta_i$ and $\eta_i=\rho\,\omega_i/\hat\beta_i$, and define $\gamma_i\triangleq \beta_i/\hat\beta_i$ and $\bar\gamma=\sum_i \omega_i\gamma_i$.
Then the self-centered estimator satisfies
\begin{equation}\label{eq:bias-formula}
\mathbb{E}[r_c]-\bar q_c \;=\; \sum_i \omega_i(\gamma_i-1)\Big(q_{i,c}-\tfrac{1}{K}\Big).
\end{equation}
In particular, if $|\gamma_i-1|\le \delta$ for all $i$, then
$\|\mathbb{E}[r]-\bar q\|_2 \le \delta \sqrt{\tfrac{K-1}{K}}$.
\end{proposition}

\begin{proof}
With mismatch, $\mathbb{E}[Y_c]=SM\rho \sum_i \omega_i\gamma_i q_{i,c}+SM\sigma_N^2$ and $\mathbb{E}[\bar Y]=SM\rho \bar\gamma/K+SM\sigma_N^2$.
Plugging into~\eqref{eq:self-centering-estimator} yields $\mathbb{E}[r_c]=\sum_i \omega_i\gamma_i q_{i,c} + (1-\bar\gamma)/K$.
Subtract $\bar q_c=\sum_i \omega_i q_{i,c}$ to obtain~\eqref{eq:bias-formula}.
The $\ell_2$ bound follows from Cauchy–Schwarz and that $\sum_c (q_{i,c}-1/K)=0$ with maximal norm $\sqrt{(K-1)/K}$ on the simplex.
\end{proof}

Our design philosophy deliberately prioritizes unbiasedness over minimum single-round MSE.
This choice is motivated by the unique demands of iterative learning algorithms like FD.

\begin{lemma}[Bias-Induced Fixed-Point Shift]
Consider an FD update step driven by a loss function that penalizes the KL-divergence between the model's output and the global average $\mathbf{q}^\star$.
If the aggregator consistently returns a biased estimate $\hat{\mathbf{q}} = \mathbf{q}^\star + \mathbf{b}$ (where $\mathbf{b}$ is a persistent bias vector), the stationary point of the learning process will be shifted.
The final learned model $\theta^*$ will converge to a solution that satisfies $\nabla\mathcal{L}(\theta^*) + \gamma \mathbf{J}(\theta^*) \mathbf{b} = 0$, representing an $O(\|\mathbf{b}\|)$ deviation from the true optimum.
\end{lemma}

An unbiased aggregator, by contrast, ensures that any error in a given round is zero-mean noise ($\mathbb{E}[\boldsymbol{\varepsilon}]=\boldsymbol{0}$), which can be averaged out over the course of training.
It does not introduce a systematic drift.

\paragraph{Pathloss mismatch and a ratio-estimator variant.}
If $\beta_i$ mismatch introduces scaling errors, periodic coarse calibration bounds the bias.
Alternatively, a pilot-free ratio estimator cancels unknown large-scale factors: let $Y_c^{(k)}$ be received energy for class $k$ on subchannel $c$, then
$
\tilde q_c^{(k)} = \frac{Y_c^{(k)}}{\sum_{d=1}^K Y_d^{(k)}},
$
which cancels common $\beta_i$ factors;
its bias is second-order in noise via a delta-method correction.
\subsection{Preservation of Dark Knowledge}
The bias in MSE-optimal estimators often manifests as a "shrinkage" effect, where the estimate is pulled towards the origin to dampen noise.
This is particularly harmful for FD because it can systematically destroy \textbf{dark knowledge}.
The small, non-maximal probabilities in a soft-label vector represent crucial information about class similarity.
These low-energy signals are the most vulnerable to being mistaken for noise and suppressed by a shrinkage estimator.
Our unbiased, noncoherent scheme does not perform such filtering. It faithfully transmits a noisy but complete version of the soft labels, preserving the faint-but-vital dark knowledge in expectation and trusting the learning algorithm to average out the zero-mean noise over time.
\section{Comparative MSE and Crossover Analysis}
Let $T_{\mathrm{coh}}$ denote per-round pilot/feedback overhead (in REs) for coherent OTA--FD,
and let the total round budget be $B$ REs.
A fair comparison holds $B$ fixed.
\begin{itemize}[leftmargin=*]
\item \textbf{Coherent OTA--FD:} usable aggregation REs $\approx B-T_{\mathrm{coh}}$;
aggregation MSE scales as $O\!\big(1/M\big)$ per RE but must be amortized over fewer repetitions.
\item \textbf{SCENE:} usable aggregation REs $\approx B$ with $S \approx B/K$ repetitions;
aggregation MSE $=O(1/(SM))$ without CSI overhead.
\end{itemize}
\noindent A crossover occurs when $T_{\mathrm{coh}} \gtrsim S$, i.e., under short coherence or large $M$,
SCENE invests more wall-clock in averaging rather than pilots.
For large $K$, top-$T$ logits or low-dimensional projections
reduce REs while introducing a controllable modeling bias that is reflected in our convergence term $O(\lambda^2\|b\|^2)$.
A persistent bias $b$ shifts the stationary point by $\Delta\theta\!\approx\!-\eta\lambda(\nabla^2F)^{-1}G b$, magnitude $O(\lambda\|b\|)$. 
SCENE’s zero bias thus preserves dark knowledge.
Under $L$-smoothness and step $\eta_t=\eta_0/\sqrt{t+1}$,
\begin{equation}
\min_{t<T}\mathbb{E}\|\nabla F(\theta_t)\|^2
=O(1/\sqrt{T})+O(\lambda^2\|b\|^2)+O(\lambda^2/(SM)).
\end{equation}

\section{Practical Considerations}
\medskip\noindent\textbf{Correlation-aware variance.}
If repetitions (time) have autocorrelation $\rho_t(\tau)$ and antennas have spatial correlation $\rho_s(d)$,
then the effective sample sizes
\[
S_{\rm eff} \;\triangleq\; \frac{S}{1+2\sum_{\tau\ge 1}\rho_t(\tau)},\qquad
M_{\rm eff} \;\triangleq\; \frac{M}{1+2\sum_{d\ge 1}\rho_s(d)}
\]
yield the approximation $\mathrm{Var}(r_c)\approx \frac{c(\bar q,\mathrm{SNR})}{S_{\rm eff}M_{\rm eff}}$.
This follows by standard Bartlett–Newey variance inflation factors for weakly dependent sums.

\textbf{Timing/CFO tolerance.} Noncoherent aggregation still requires slot-level alignment so that energies add on the same RE.
Single-carrier constant-envelope waveforms with modest guard intervals reduce sensitivity to CFO/time offset while keeping PAPR low.
\noindent\textbf{Spatial/temporal correlation.} The $O(1/(SM))$ variance law assumes independence across repetitions and antennas.
With correlation coefficients $\rho_t$ (time) and $\rho_s$ (space), one can replace $SM$ by an effective $S_{\mathrm{eff}}M_{\mathrm{eff}}$
(e.g., $M_{\mathrm{eff}}=M(1-\rho_s)$, $S_{\mathrm{eff}}=S(1-\rho_t)$ for a first-order correction).
Frequency hopping across reps helps decorrelate.

\noindent\textbf{Dynamic range and quantization.} Peaky soft labels can cause tiny entries to be
noise/quantization-limited.
A small floor $\varepsilon$ or temperature smoothing mitigates dynamic-range issues at negligible accuracy cost.

\noindent\textbf{Robustness and privacy.} Self-centering plus post-adjustment (projection/renorm) gives natural defenses against outliers.
When needed, per-device energy caps and trimmed-mean on energies can bound the effect of adversarial scaling.
Small DP noise added \emph{before} energy mapping composes well with the noncoherent sum.

\noindent\textbf{Hardware Feasibility}
Beyond protocol overhead, our scheme offers a significant advantage in hardware feasibility that is critical for energy-constrained edge devices.
By design, the total transmit energy per device is constant ($\sum_c E_{i,c}^{(k)} = \eta_i^{(k)}$), allowing for the use of constant-envelope waveforms with a Peak-to-Average Power Ratio (PAPR) approaching unity.
This is highly favorable for hardware design, as it enables the use of power-efficient, non-linear power amplifiers (PAs) operating near saturation with minimal back-off.
In sharp contrast, coherent schemes, particularly those based on multicarrier modulation like OFDM, often exhibit high PAPR with crest factors exceeding 10 dB.
Such signals necessitate the use of expensive, linear PAs that require significant power back-off to avoid distortion, fundamentally reducing energy efficiency.
This hardware-friendliness is not a minor detail; it materially impacts the energy consumption per round and reinforces the practical superiority of our 'simplicity by design' approach in real-world edge settings.
\begin{table}[ht] 
\centering 
\caption{Overhead and Crossover}
\label{tab:overhead} 
\begin{tabular}{lccccc}
\toprule
Scheme & Pilot Cost & REs/round & Var.
& PAPR & CSI Need \\
\midrule
Coherent OTA--FD & High & $K+T_{\text{pilot}}$ & $O(1/M)$ & $>10$~dB & Yes \\
SCENE & None & $KS$ & $O(1/(SM))$ & $\approx 1$ & No \\
Ratio--SCENE & None & $K+1$ & $O(1/(SM))$ & $\approx 1$ & No \\
\bottomrule
\end{tabular}
\end{table}

\section{Coherent vs.\ Noncoherent: A Simple Crossover Model}
Assume a per-round RE budget $B$ and $K$ classes. Coherent OTA--FD requires $P$ REs for pilots/feedback/CSI acquisition, so $S_{\rm coh}=(B-P)/K$ usable repeats;
SCENE uses $S_{\rm nc}=B/K$ (no UL pilots). Suppose $c_{\rm coh}$ and $c_{\rm nc}$ are scheme-dependent MSE constants that absorb all non-asymptotic factors (e.g., estimator efficiency, noise enhancement, channel-estimation/beamforming errors for coherent OTA--FD, and phase-loss/energy-detection variance for noncoherent SCENE) beyond the common $\Theta(1/(MS))$ averaging scaling. Therefore, the round-MSEs scale as
$\frac{c_{\rm coh}}{M S_{\rm coh}}$ and $\frac{c_{\rm nc}}{M S_{\rm nc}}$, respectively.
Hence, SCENE has lower MSE iff
\begin{equation}
\frac{c_{\rm nc}}{B} \;\le\; \frac{c_{\rm coh}}{B-P}
\quad\Longleftrightarrow\quad
P \;\ge\; \Big(1 - \frac{c_{\rm coh}}{c_{\rm nc}}\Big) B.
\end{equation}
Empirically, $c_{\rm nc}>c_{\rm coh}$ because noncoherent detection discards phase, but when $P$ is a significant fraction of $B$ (e.g., short coherence/high mobility),
the increase in usable repeats can outweigh the constant gap. Indeed, SCENE buys itself extra averaging repetitions by not paying the pilot tax. The threshold
\[
P_{\text{threshold}}=\left(1-\frac{c_{\text{coh}}}{c_{\text{nc}}}\right)B
\]
is the pilot-tax level above which coherent loses because it wastes too much of the round just acquiring CSI.

\subsection*{SCENE in the Design Space of FL}
\emph{Versus digital FD.} State-of-the-art digital FD achieves excellent bit-efficiency via quantization and entropy coding. SCENE’s motivation is different: one-shot analog superposition yields ultra-low latency and device simplicity (constant power, constant envelope), which is attractive at large scale. Our goal is not to beat digital in bits per label, but to furnish a robust analog primitive when latency and hardware constraints dominate.

\emph{Versus hybrid coherent schemes.} A middle ground is possible: infrequent pilots to capture partial coherent gain. This can reduce MSE constants but reintroduces broadcast pilot mechanisms, client-side phase tracking, and sensitivity to oscillator drift. Residual phase/amplitude errors may elevate the noise floor and, under persistent miscalibration, introduce bias. SCENE deliberately avoids this complexity: it is unbiased by construction, insensitive to fast phase variation, and competitive whenever the pilot fraction is non-negligible (short coherence or many users).

\section{Simulation Results}

We consider a federated distillation system with a single parameter server (PS) and several edge devices (clients). We evaluate on MNIST (60{,}000 train / 10{,}000 test, 10 classes). We fix the number of clients to $N=100$ and split the 60k training images into a \emph{shared unlabeled open set} of size $I_o$ and a \emph{labeled private set} of size $I_p$ (with $I_o + I_p \le 60{,}000$) \cite{itahara2020dsfl}. We fix the shared unlabeled open set to have size $I_o = 20{,}000$ and the labeled private set to have size $I_p = 40{,}000$ (so that $I_o + I_p = 60{,}000$) \cite{itahara2020dsfl}. For IID experiments, the private set is shuffled and evenly partitioned so that each client holds $I_p/N$ labeled examples. All reported experimental results are averaged over 20 independent trials.

\paragraph{Model and training.} We use the MNIST CNN (similar to \cite{itahara2020dsfl}) implemented in our codebase for both client and server models. The network consists of two $5{\times}5$ convolutional layers with 32 and 64 channels, respectively, each followed by ReLU and $2{\times}2$ max pooling. The resulting feature map is flattened and passed through two fully-connected layers: 512 units with ReLU, followed by a 10-way linear classifier. The softmax is applied to convert the outputs to probabilities. This architecture has $1{,}663{,}370$ trainable parameters. $U$ denotes the unlabeled budget, i.e., the number of unlabeled samples selected from the shared pool $I_o$ and used for one-shot distillation.

For the one shot protocol, clients first perform supervised pretraining for a fixed number of local epochs on labeled data using mini-batch SGD (batch size 32, learning rate 0.005, momentum 0.9, weight decay $5{\times}10^{-4}$). Next, the server runs a one-shot distillation stage: for each unlabeled budget $U$, clients provide class-probability predictions on the selected unlabeled subset, which are aggregated either over-the-air (OTA) or via a noise-free weighted average (Plain). The server then distills for a fixed number of epochs using unlabeled mini-batches of size 32 and evaluates the global model after distillation.
  Distillation is carried out primarily at the server by Kullback-Leibler divergence (KLD) as the loss function: once per round on the shared unlabeled batch, while clients do not perform additional local distillation epochs beyond their supervised updates. The client pretraining stage and the server distillation stage are both run for 20 epochs.

\paragraph{Wireless channel and path-loss.}
For the OTA experiments, we simulate a broadband uplink with noncoherent combining at the PS. Each client $i$ experiences a large-scale gain
\[
\beta_i = d_i^{-\alpha} 10^{X_i/10},
\]
where $d_i$ is the distance from client $i$ to the PS, drawn uniformly in the interval $[5,50]$ meters, $\alpha = 3.5$ is the path-loss exponent, and $X_i$ is a zero-mean Gaussian random variable with standard deviation $8$ dB modeling lognormal shadowing. The resulting gains are normalized such that their empirical mean equals one. Per-device power constraints are modeled by assigning each client a maximum transmit power uniformly drawn from a fixed interval $P_i \sim \mathcal{U}[0.5, 1.5].$ ; all devices employ the same mapping from soft-label probabilities to transmit energies, ensuring constant total transmit energy per device per round. Small-scale fading is modeled as independent Rayleigh fading across devices, antennas, and repetitions. The PS is equipped with $M$ antennas and each OTA burst is repeated $S$ times, so that each class-energy measurement is averaged over $SM$ independent observations.

The noise variance at the receiver is chosen to realize a prescribed target signal-to-noise ratio (SNR) per resource element. Given the average received signal energy implied by the soft-label energies and path-loss, the noise power is set so that the effective per-resource-element SNR equals either $5$ dB or $10$ dB, depending on the scenario. Clients training and test accuracies are 100 percent and 92.18 percent, respectively. Clients train/test accuracies are reported before the OTA distillation stage (they are the local pretraining accuracies). Indeed, SNR affects the OTA aggregation/distillation performance, not the clients’ supervised pretraining metrics—so the client averages remain the same across SNR sweeps.

In our experiments, the \emph{Plain} baseline aggregates client predictions by a noise-free, data-size–weighted arithmetic mean of the class-probability vectors on the simplex, i.e., it reflects ideal server-side averaging without OTA impairments.

\paragraph{SCENE aggregation and operating points.}
Devices map their soft-label vectors to nonnegative energies and transmit concurrently over $K$ orthogonal resource elements. The PS aggregates the received magnitudes squared over $S$ repetitions and $M$ antennas and applies the SCENE self-centering estimator, followed by a light projection back onto the probability simplex (clipping and renormalization). 
In all cases, the learning-rate schedule, batch sizes, number of local epochs, and distillation protocol are kept fixed; only the aggregation mechanism, SNR, and diversity parameters $(S,M)$ are varied.

\begin{figure}[t]
    \centering
    \includegraphics[width=\linewidth]{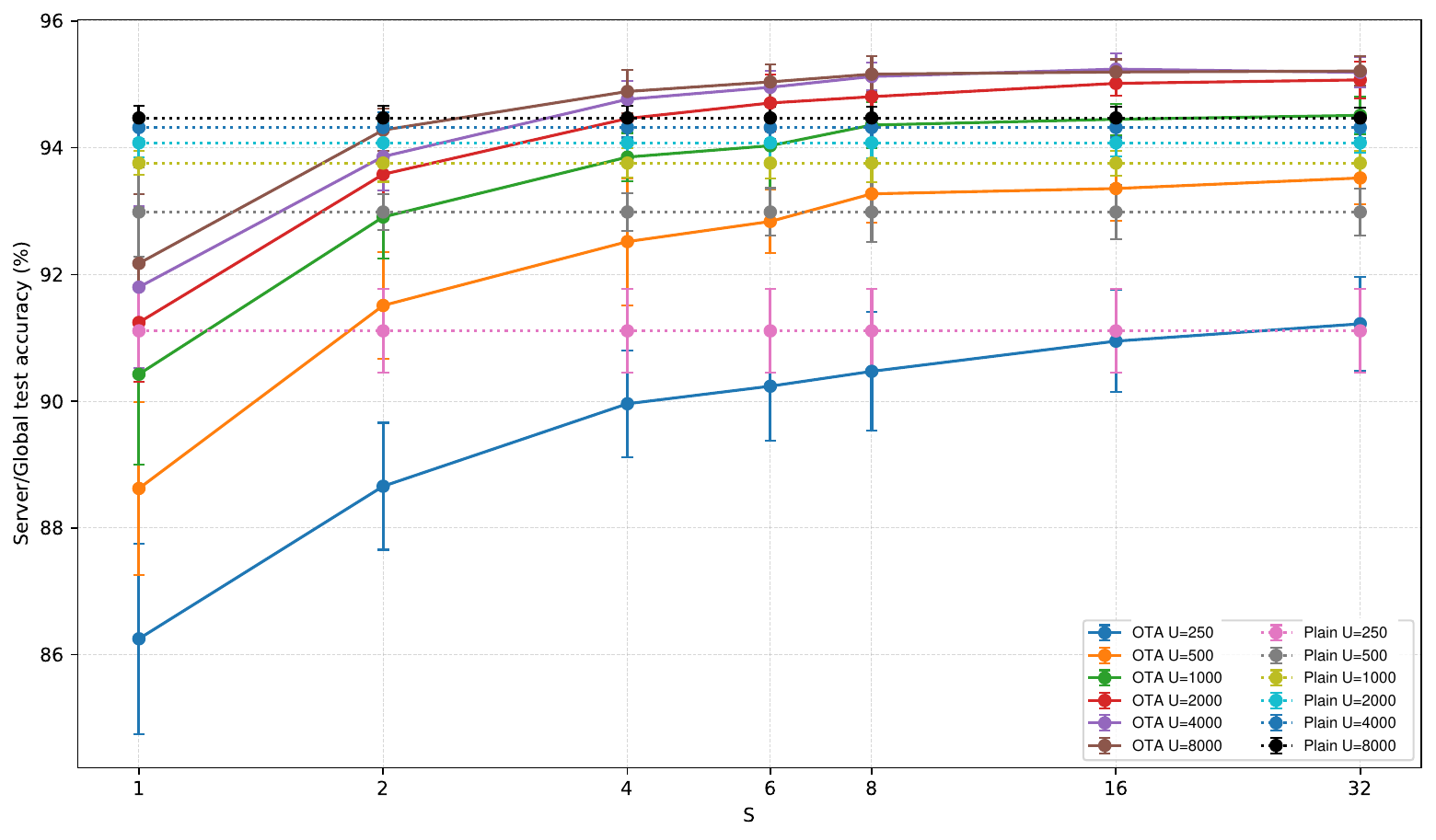} 
    \caption{S versus server/global test accuracy (percent) at SNR = 5 dB.}
    \label{fig:Server_S_5db}
\end{figure}

\begin{figure}[t]
    \centering
    \includegraphics[width=\linewidth]{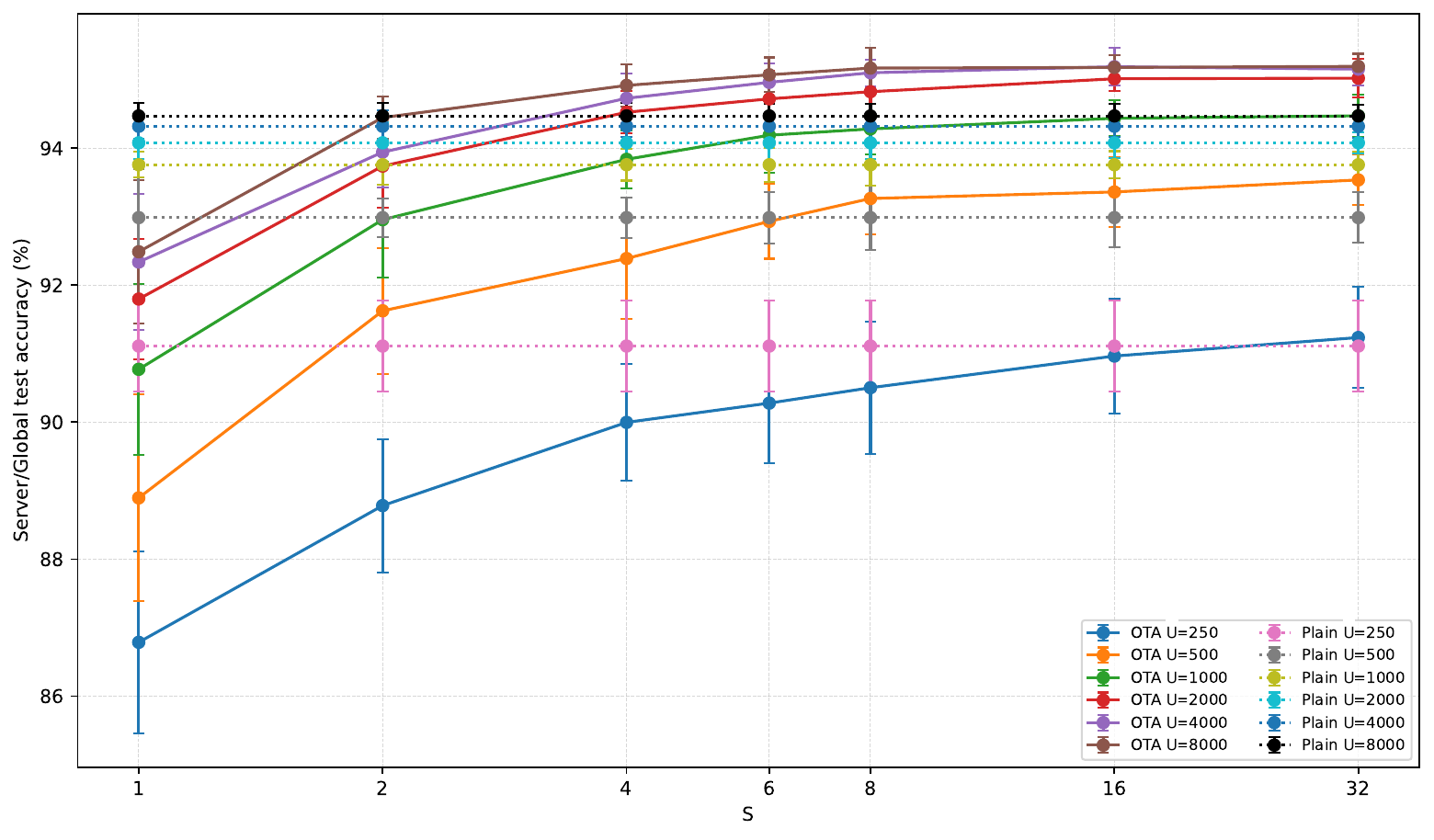} 
    \caption{S versus server/global test accuracy (percent) at SNR = 10 dB.}
    \label{fig:Server_S_10db}
\end{figure}

Indeed, the server forms soft targets for distillation by aggregating client-predicted class probabilities over the wireless channel. For OTA aggregation, each unlabeled sample can be transmitted multiple times, i.e., repeated $S$ times, which effectively averages channel impairments across repetitions. Consequently, considering $M = 1$, the total OTA communication/airtime cost scales with the product of the unlabeled batch size and the repetition factor, $B \propto U S$. Under a fixed budget $B$, there is an inherent trade-off: increasing the unlabeled batch size $U$ improves coverage/diversity for knowledge transfer, whereas increasing the repetition factor $S$ improves the reliability of the aggregated soft labels by reducing OTA-induced distortion.

Empirically, as shown in Figures \ref{fig:Server_S_5db} and \ref{fig:Server_S_10db} this trade-off produces a clear ``sweet spot'' in $(U,S)$: allocating some budget to repetition can be strictly more effective than spending the entire budget on sending a larger unlabeled batch only once. Intuitively, when $S$ is too small, the OTA aggregation noise degrades the soft targets and limits distillation quality even if many unlabeled samples are used. In contrast, a moderate number of repetitions substantially stabilizes the received soft labels and yields a marked accuracy improvement, often outperforming the ``send-once'' strategy at the same total budget. Beyond this regime, however, excessive repetition exhibits diminishing returns because the reduced $U$ eventually limits the diversity of unlabeled data, leading to saturation and, in some cases, a decline in performance.

This behavior is consistent across channel conditions: at lower SNR, as in Figure \ref{fig:Server_S_5db}, the benefit of repetition is amplified (since OTA impairments are stronger), while at higher SNR, as in Figure \ref{fig:Server_S_10db}, the optimum shifts slightly but remains in the moderate-repetition regime. Overall, these results support the design principle that, for a fixed OTA budget, \emph{moderate repetition is typically preferable to a single transmission of a larger unlabeled batch}, indicating a robust sweet-spot allocation that balances unlabeled coverage and OTA reliability.

At this stage, we hypothesize that when OTA outperforms the ideal noise-free aggregation, it is because residual OTA perturbations can act as an implicit regularizer during distillation, improving the neural network’s generalization.

\paragraph{Why the $S$--$M$ curves nearly coincide when $SM$ is fixed.}
The close overlap of the OTA curves in Figure~\ref{fig:Diff_comb_5db} for different $(S,M)$ pairs with the same product $SM=16$ is consistent with the way OTA noise enters the aggregation in our implementation. In the code, $M$ (the number of receive antennas/branches) provides \emph{spatial averaging} of the OTA distortion, while $S$ (the number of repeated transmissions) provides \emph{temporal averaging}. Both mechanisms reduce the effective variance of the aggregation error through averaging across approximately independent realizations. As a result, for fixed total diversity order $SM$, the received aggregate used to generate distillation targets has a similar effective signal-to-distortion level, yielding comparable quality soft labels and hence similar server test accuracy as a function of $U$. Residual gaps at small $U$ can be attributed to finite-sample effects and non-idealities (e.g., mild dependence between repetitions/branches and nonlinearities from probability normalization), but overall the observed invariance indicates that, in this regime, performance is primarily governed by the \emph{total averaging budget} $SM$ rather than how it is split between repetitions and antennas.

By exchanging soft labels (class-probability vectors) on a shared unlabeled batch instead of full model parameters or gradients, federated distillation reduces the uplink payload from model-size $O(P)$ (with $P$ parameters) to label-size $O(K)$ (with $K$ classes), yielding substantial communication savings and a representation that is naturally compatible with over-the-air aggregation. Moreover, it enables heterogeneous model architectures across clients and the server, since aggregation operates on predicted probability vectors rather than model parameters.

\begin{figure}[t]
    \centering
    \includegraphics[width=\linewidth]{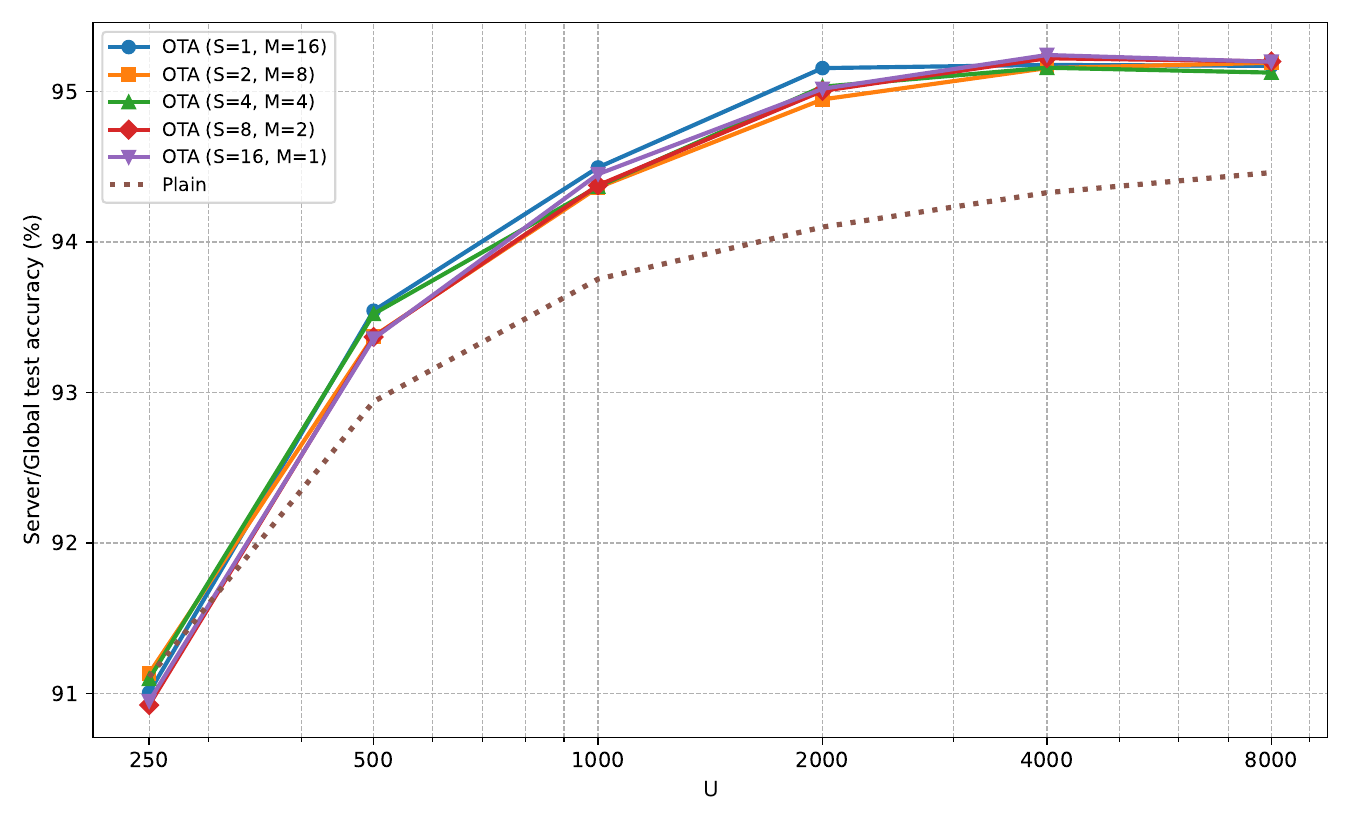} 
    \caption{different combinations of (S,M)={(1,16),(16,1),(2,8),(8,2),(4,4)} for 100 clients, one-shot, average of 20 trials at SNR = 5 dB.}
    \label{fig:Diff_comb_5db}
\end{figure}

\section{Conclusion}
We have proposed a simple, robust, and pilot-free noncoherent OTA aggregation scheme for Federated Distillation.
By mapping soft-label vectors directly to transmit energies, our uplink is invariant to channel phase and CFO;
combined with a practical \emph{self-centering} estimator, this yields a \emph{pilot-free} implementation.
The constant-envelope signaling (PAPR~$\approx 1$) further makes the approach \emph{hardware-friendly} for edge devices.
Crucially, we prioritize \emph{unbiasedness} over minimal single-round MSE. Whereas coherent LMMSE-style designs reduce variance via shrinkage—risking systematic attenuation of the probability tail and the loss of ``dark knowledge’’—our estimator preserves the target distribution \emph{in expectation}, aligning better with the needs of iterative learning and helping retain dark knowledge.
Our analysis frames a fundamental trade-off between CSI acquisition and bandwidth, identifying short-coherence/high-mobility and hardware-constrained regimes as key operating points where the proposed noncoherent design excels, delivering superior test accuracy versus wall-clock time despite a higher per-round MSE.
\bibliographystyle{IEEEtran}

\bibliography{nc_softlabel_refs}

\end{document}